\documentclass[review]{elsarticle}

\usepackage{microtype}
\usepackage{graphicx}
\usepackage{bbm}
\usepackage{enumerate}
\usepackage{booktabs}
\usepackage{tabu}
\usepackage[ruled,lined,boxed,linesnumbered]{algorithm2e}
\usepackage{amsfonts,amsthm,amsmath,amssymb,mathtools}
\usepackage{xcolor}
\usepackage{xspace}
\usepackage{float}
\usepackage{mathtools}
\usepackage{tabularx}
\usepackage{relsize}

\newtheorem*{Tracking Set Condition}{Tracking Set Condition}
\newtheorem{Reduction Rule}{Reduction Rule}
\newtheorem{Observation}{Observation}

\newtheorem{theorem}{Theorem}
\newtheorem{lemma}{Lemma}
\newtheorem{corollary}{Corollary}

\newcommand{\prob}[3]{
  \vspace{2mm}
\noindent\fbox{
  \begin{minipage}{0.96\textwidth}
  \vspace{-1mm}
  \begin{tabular*}{\textwidth}{@{\extracolsep{\fill}}lr} \textsc{#1} 
  \vspace{-1mm} 
\\ \end{tabular*}
  {\bf{Input:}} #2  
  
  {\bf{Output:}} #3
  \end{minipage}
  }
  \vspace{2mm}
}

\newcommand{\pprob}[3]{
  \vspace{2mm}
\noindent\fbox{
  \begin{minipage}{0.96\textwidth}
  \vspace{-1mm}
  \begin{tabular*}{\textwidth}{@{\extracolsep{\fill}}lr} \textsc{#1} 
  \vspace{-1mm} 
\\ \end{tabular*}
  {\bf{Input:}} #2  
  
  {\bf{Question:}} #3
  \end{minipage}
  }
  \vspace{2mm}
}

\usepackage{microtype}
\usepackage{xspace}

\newcommand{\stgraph}{{$s$-$t$ graph}\xspace}
\newcommand{\stpath}{{$s$-$t$ path}\xspace}
\newcommand{\sstpath}{{shortest $s$-$t$ path}\xspace}
\newcommand{\sstpaths}{{shortest $s$-$t$ paths}\xspace}
\newcommand{\stpaths}{{$s$-$t$ paths}\xspace}

\newcommand{\trs}{{tracking set}\xspace}

\newcommand{\hs}{{hitting set}\xspace}
\newcommand{\hsp}{{\sc Hitting Set}\xspace}

\newcommand{\dhs}{{\sc $d$-Hitting Set}\xspace}
\newcommand{\tsc}{\textit{tracking set condition}\xspace}
\newcommand{\tp}{{\sc Tracking Paths}\xspace}
\newcommand{\ddtsp}{{\sc $diam$-$d$-Tracking Shortest Paths}\xspace}

\newcommand{\tsp}{{\sc Tracking Shortest Paths}\xspace}

\newcommand{\tss}{{\sc Tracking Set System}\xspace}
\newcommand{\dts}{{\sc $d$-Tracking Set}\xspace}
\newcommand{\tcpr}{{\sc Test Cover}\xspace}
\newcommand{\tc}{{test cover}\xspace}
\newcommand{\fvs}{{FVS}\xspace}
\newcommand{\tpdag}{{\sc Tracking Paths in DAGs}\xspace}
\newcommand{\dc}{{\sc Discriminating Code}\xspace}
\newcommand{\ivcp}{{\sc Identifying Vertex Cover}\xspace}

\newcommand{\NP}{\text{\normalfont  NP}\xspace}
\newcommand{\NPH}{\text{\normalfont  NP}-hard\xspace}

\newcommand{\FPT}{\text{\normalfont FPT}\xspace}

\newcommand{\W}[1][xxxx]{\text{\normalfont W[#1]}\xspace}

\newcommand{\APX}{{\sc APX}\xspace}

\newcommand{\Oh}{\mathcal{O}}

\journal{Journal of Theoretical Computer Science}

\begin{document}

\begin{frontmatter}

\title{Fixed-parameter tractable algorithms for Tracking Shortest Paths \tnoteref{mytitlenote}}
\tnotetext[mytitlenote]{Under review in a journal. A preliminary version of the paper appeared in the proceedings of CALDAM 2018~\cite{caldam2018}}


\author{Aritra Banik}
\address{National Institute of Science Education and Research Bhubaneswar, HBNI, India.}
\ead{aritra@niser.ac.in}

\author{Pratibha Choudhary}
\address{Indian Institute of Technology Jodhpur, Jodhpur, India.}
\ead{pratibhac247@gmail.com}

\author{Venkatesh Raman}
\address{The Institute of Mathematical Sciences, HBNI, Chennai, India.}
\ead{vraman@imsc.res.in}

\author{Saket Saurabh}
\address{The Institute of Mathematical Sciences, HBNI, Chennai, India.}
\ead{saket@imsc.res.in}

\setcounter{footnote}{0}
\begin{abstract}
We consider the parameterized complexity of the problem of tracking \sstpaths in graphs, motivated by applications in security and wireless networks. 
Given an undirected and unweighted graph with a source $s$ and a destination $t$, \tsp asks if there exists a $k$-sized subset of vertices (referred to as \textit{tracking set}) that intersects each \sstpath in a distinct set of vertices. 

We first generalize this problem for set systems, namely \tss, where given a family of subsets of a universe, we are required to find a subset of elements from the universe that has a unique intersection with each set in the family. \tss is shown to be fixed-parameter tractable due to its relation with a known problem, \tcpr. By a reduction to the well-studied $d$-hitting set problem, we give a polynomial (with respect to $k$) kernel for the case when the set sizes are bounded by $d$. This also helps in solving \tsp when the input graph diameter is bounded by $d$.

While the results for \tss show that \tsp is fixed-parameter tractable, we also give an independent algorithm  by using some preprocessing rules, resulting in an improved running time.

\end{abstract}


\end{frontmatter}


\section{Introduction and Motivation}
In this paper, we consider the parameterized complexity of the problem of tracking \textit{shortest $s$-$t$ paths} in graphs and some related versions of the problem. Given a graph with a specified source $s$ and a destination $t$, a simple path between $s$ and $t$ is referred to as an \textit{$s$-$t$ path}, and a shortest simple path between $s$ and $t$ is referred to as a \textit{shortest $s$-$t$ path}. In \tsp problems, the goal is to find a  small subset of vertices that can help uniquely identify all \sstpaths in a graph.

We start with some motivation for the problem.
Consider the security system at a large airport. As a security measure, it is required to identify the routes taken by passengers across the airport from entry to departure or from arrival to exit. A set of carefully chosen security scan points can be selected as identification points to trace the movements of passengers. A similar scenario can arise in any other secure facilities where movement of entities need to be tracked. Note that in practical scenarios, often it is resourceful to use the \sstpaths available. 

Other major application scenarios are tracking of moving objects in telecommunication networks and  road networks. The goal can be efficient and optimized tracking of objects, for the purpose of surveillance, monitoring, intruder detection, and operations management. 
The solution to the problem can then be used for reconstruction of path traced by an object in order to detect potential network flaws, to study traffic patterns of moving objects, to optimize network resources based on such patterns, and for other such network analysis based tasks.

Tracking of moving objects has been studied in the field of wireless sensor networks. See~\cite{survey} for a survey of target tracking protocols using wireless sensor networks. Some researchers have studied this with respect to power management of sensors~\cite{power}. Despite being an active area of research, a major part of this research so far has been based on heuristics.
In~\cite{banik}, the authors formalized the problem of tracking in networks as a graph theoretic problem and did a systematic study. Among other problems, they introduced the following optimization problem. 
$V(P)$ is used to denote the set of vertices in path $P$. A graph with a unique source $s$ and unique destination $t$ is called an $s$-$t$ graph.



\prob{\tsp}{An undirected \stgraph $G=(V,E)$.}{A minimum set of vertices ${T} \subseteq V$, such that for any two distinct \sstpaths $P_1$ and $P_2$ in $G$, it holds that ${T} \cap V(P_1) \neq {T} \cap V(P_2)$.}

The output set of vertices is referred to as a \textit{tracking set} and the vertices in a \textit{tracking set} are called \textit{trackers}.
In~\cite{banik}, \tsp was shown to be {\sc NP}-hard for undirected graphs and a $2$-approximate algorithm was given for the case of planar graphs. An $\alpha$-approximation algorithm for a minimization problem gives a solution that is at most $\alpha$ times the size of an optimum solution, in time polynomial in the input size.

\tsp can be generalized to the case where not just the \sstpaths, but all \stpaths in a graph need to be identified uniquely by a minimum subset of vertices. For a set of vertices $T\subseteq V$ and a path $P$, $\Pi_P(T)$ denotes the sequence in which the vertices from $V(P)\cap T$ appear in path $P$. Formally the problem of tracking all \stpaths in a graph is defined as follows.

\prob{\tp}{An \stgraph $G=(V,E)$.}{A minimum set of vertices ${T} \subseteq V$, such that for any two distinct \stpaths $P_1$ and $P_2$ in $G$, it holds that $\Pi_{P_1}({T}) \neq \Pi_{P_2}({T})$.}

For a graph $G$, solving \tp requires finding a \trs that intersects all \stpaths in a unique sequence. 
Note that if we consider only \sstpaths, a pair of vertices cannot appear in different sequence in two distinct \sstpaths, as this would mean that at least one of the paths is not a \sstpath in the graph. Hence in case of \tsp, it is sufficient to find a \trs that intersects each \sstpath in a unique set of vertices.

It has been proven that \tp is {\sc NP}-hard for undirected graphs, and admits a polynomial kernel when parameterized by $k$, the size of the \trs~\cite{latin-journal-2019}. 
In parameterized complexity, a (polynomial) \textit{kernel} is an equivalent instance whose size is a (polynomial) function of a parameter $k$, where $k$ is either the size of the output or some other integer related to the input instance such that $k$ is preferably very small compared to the input size. Later this result was improved in~\cite{quadratic} by showing the existence of an $\mathcal{O}(k^2)$ kernel for undirected graphs and an $\mathcal{O}(k)$ kernel for undirected planar graphs. See Section~\ref{subsec:fpt} for details on \FPT and kernels. 

Observe that for a graph $G$, a \trs for all \stpaths is also a \trs for all \sstpaths. However, it may be the case that $G$ does not have a \trs of size at most $k$ for all \stpaths, but it might still have a \trs of size at most $k$ for all \sstpaths. Hence the parameterized complexity of \tsp is a problem of independent interest.

The key idea behind the kernel for \tp in~\cite{latin-journal-2019},~\cite{quadratic} originates from the fact that for an undirected graph $G$, a \trs for all \stpaths is also a feedback vertex set (\fvs) for $G$. A \textit{feedback vertex set} for a graph $G$ is a set of vertices whose removal makes $G$ acyclic. However, a \trs for all \sstpaths in a graph need not be a \fvs. 
In this paper we address the parameterized complexity of \tsp along with its restricted version \ddtsp. \ddtsp requires finding a \trs for distinguishing between \sstpaths in a graph whose diameter is restricted to $d$. We show that \ddtsp is \NP-hard and \FPT.

We first study a combinatorial version of \tsp, which is \tss, in Section~\ref{sec:setsystem}.
A set system is a pair $\mathcal{P}=\{X,\mathcal{S}\}$, where $X$ is a finite set and $\mathcal{S}$ is a family of subsets of $X$.
For a set system, a \trs is a set of elements that has a unique intersection with each of the subsets in the family. \tss is formally defined as follows.


\prob{\tss}{A set system $\mathcal{P}=\{X,\mathcal{S}\}$.}{A minimum cardinality set ${T} \subseteq X$, such that for any two distinct $S_i,S_j\in S$, it holds that $S_i\cap T\neq S_j\cap T$.}

Here the elements in a \trs are referred as \textit{trackers}. 


\tss has some resemblance to the well known \hsp problem. For a set system $(U,\mathcal{F})$ comprising of a finite universe $U$ and a collection $\mathcal{F}$ of subsets of $U$, a {\em hitting set} is a set $H \subseteq U$ that has a non-empty intersection with each set in $\mathcal{F}$, and the optimization version of \hsp requires finding a minimum cardinality \hs. Observe that while \hs is a set of elements that is required to have a non-empty intersection with each of the sets in the set system family, a \trs is required to have a unique intersection with each of the sets in the family.
\hsp was one of Karp's original NP-complete problems~\cite{karp}.


We first study \dts, which is a restricted version of \tss where the size of each subset in the family is restricted to $d$. We show \dts to be \NPH by showing a connection with the problems \ivcp and \textsc{Packing}~\cite{henning-identifying-vertex-covers},~\cite{chain-packing}. We then give a \textit{compression} for \dts by showing a reduction from \dts to the \dhs problem. \dhs is a restricted version of \hsp where the set sizes in the family are restricted to $d$. 
\textit{Compression} of a parameterized problem $X$ into a problem $Y$ is an algorithm that takes as input an instance $(x,k)$ of $X$, works in time polynomial in $|x|+k$, and returns a problem instance $y$, such that $|y|\leq p(k)$ for some polynomial $p(\cdot)$, and, $y$ is a YES instance of $Y$ if and only if $(x,k)$ is a YES instance of $X$. Since \dhs is in \NP and \dts is \NP-hard, there exists a polynomial reduction from \dhs to \dts as well. This gives a kernelization result for \dts.
While the reduction still works for the unrestricted version, it does not help to resolve the parameterized complexity of \tss when the set sizes are unrestricted as general \hsp is known to be hard for the parameterized complexity class {\sc W}[2]~\cite{book}.

\tss is known to be related to \tcpr~\cite{bazgan}. \tcpr requires finding a subfamily of sets in a set system, that can help identify each element in the universe uniquely by inclusion. Using known results about \tcpr~\cite{Crowston2012}, we show that the size of a \trs for a set system with $n$ elements and $m$ sets is at least $\lceil \lg m \rceil$\footnote{We use $\lg$ to denote logarithm to the base $2$}. This, along with some reduction rules, leads to the result that the problem of determining whether a given set system has a \trs of size at most $k$ has a \FPT algorithm running in time\footnote{$\mathcal{O}^*$ notation ignores the polynomial factors in terms of the size of input $n$} $\mathcal{O}^*(2^{k2^k})$.

We then consider other natural parameterizations of \tss and give \FPT algorithms and hardness results that follow from the equivalence to \tcpr.


In Section~\ref{sec:tssp}, we consider the parameterized complexity of \tsp problem. We study \tsp along with a restricted version of it i.e. \ddtsp problem. 
\ddtsp requires finding a \trs for \sstpaths when the input graph has \textit{diameter} at most $d$. Using results from Section~\ref{sec:setsystem} and~\cite{banik}, we first prove that both these problems are \NPH and admit \FPT algorithms. Then in Section~\ref{subsec:improvedfpt}, we introduce the \tpdag problem which requires finding a \trs for all (directed) \stpaths in a directed acyclic graph (DAG). We give an improved fixed-parameter tractable algorithm for \tsp by first reducing it to \tpdag, and then giving a kernel for \tpdag.

The following table gives a summary of our results in this paper.

\noindent
\begin{tabular}{ |l|l|l|l| } 
\hline
\textbf{Problem} & \textbf{Kernel} & \textbf{FPT} & \textbf{Section} \\
 \hline
 \hline
\dts & Polynomial &  $\Oh^*(c^k)$ & \ref{subsec:d-tracking-set} \\
\tss  & $\Oh(2^{2^k})$ &  $\Oh^*(2^{k2^k})$ & \ref{subsec:tracking-set-systems} \\
\ddtsp & Polynomial & $\Oh^*(c^k)$ & \ref{subsec:diam-d-tp} \\
\tsp & $\Oh(2^{k})$  & $\Oh^*(2^{k^2+3k})$ & \ref{subsec:improvedfpt} \\
\tpdag & $\Oh(2^{k})$  & $\Oh^*(2^{k^2+3k})$ & \ref{subsec:improvedfpt} \\
 \hline
\end{tabular}

Polynomial indicates polynomial in $k$ for a fixed $d$, and $c$ is a function polynomial with respect to $d$.

\subsection{Related Work}

\tss has been studied earlier under the problem name \textsc{Distinguishing Transversals in Hypergraphs}~\cite{henning-cubic}. 
Some closely related graph theoretic problems are \dc~\cite{discriminating-codes} and \textsc{Identifying Codes}~\cite{id-codes-1},~\cite{id-codes-2},~\cite{julien-thesis}. 
\textsc{Distinguishing Transversals} when restricted to 2-uniform hypergraphs is equivalent to \textsc{Identifying Vertex Cover}, which is the problem of finding a set of vertices $V'\subseteq V$ for a graph $G=(V,E)$, such that $a\cap V'\neq b\cap V'$, for a pair of distinct edges $a,b\in E$. Henning and Yeo give some bounds for the size of an output in~\cite{henning-identifying-vertex-covers} and~\cite{henning-cubic} for \textsc{Identifying Vertex Cover} and \textsc{Distinguishing Transversal}, respectively.

Recently Eppstein et al. proved \tp in planar graphs to be \NP-hard and gave a $4$-approximation algorithm for the same~\cite{eppstein}. In~\cite{guido-cubic}, Bil{\`o} et al. show that \textsc{Tracking Shortest Paths} is NP-hard for cubic planar graphs in case of multiple source-destination pairs, and give an \textsc{FPT} algorithm parameterized by the number of vertices equidistant from the source or destination. Further, some polynomial time algorithms have been given to solve \tp in some restricted classes of graphs~\cite{iwoca,polytime-arxiv}.

\section{Preliminaries}
\label{sec:prelim}


Throughout this paper, we assume that each graph is an $s$-$t$ graph with $s$ and $t$ already given to us. $V(G)$ denotes the vertex set of graph $G$ and $E(G)$ denotes the edges whose both endpoints belong to $V(G)$. We use DAG to denote directed acyclic graph. For vertices $u,v\in V(G)$ where $G$ is an undirected graph, $uv \in E(G)$ denotes an edge between vertices $u$ and $v$. For vertices $a,b\in V(G)$ where $G$ is directed graph, $(a,b) \in E(G)$ denotes a directed edge between vertices $a$ and $b$, oriented from $a$ towards $b$. Given a graph $G=(V,E)$, $G-e$ denotes the graph induced by removing the edge $e\in E$ from $G$, i.e. $G(V,E\setminus e)$. For a vertex $v\in V(G)$, neighborhood of $v$ is denoted by $N(v)$, and $N(v)=\{u \mid uv\in E(G)\}$. The degree of a vertex $v$ is denoted by $deg(v)=|N(v)|$. $N^+(v)$ denotes the set of out-neighbors of vertex $v$ i.e. $N^+(v) = \{u \mid (v,u)\in E(G)\}$ and $N^-(v)$ denotes the set of in-neighbors of $v$ i.e. $N^-(v)=\{w\mid (w,v)\in E(G)\}$. The out-degree of a vertex $v$ is equal to $|N^+(v)|$ and is denoted by $deg^+(v)$ and in-degree is equal to $|N^-(v)|$ and is denoted by $deg^-(v)$. For a vertex $v$ in a directed graph, the degree of $v$ is $deg(v)=deg^+(v)+deg^-(v)$ and the neighborhood of $v$ is $N(v)=N^+(v)\cup N^-(v)$. \textit{Short-circuiting} a vertex of degree two means deleting the vertex and introducing an edge between its neighbors.

A path is a sequence of vertices where subsequent vertices are connected by an edge. We only consider simple paths in this paper i.e. paths that do not repeat vertices. $V(P)$ is used to denote the vertex set of path $P$. For vertices $a,b\in V$, an $a$-$b$ path means a path between vertices $a$ and $b$.
If there exists a path $P_1$ between vertices $u$ and $v$, and there exists another path $P_2$ between vertices $v$ and $w$, we use $P_1\cdot P_2$ to denote the path between $u$ and $w$ obtained by concatenation of paths $P_1$ and $P_2$ at $v$. 
The length of a path is equal to the number of edges in that path. The $distance$ between two vertices $x,y\in V(G)$, denoted by $dis(x,y)$, is the length of the shortest $x$-$y$ path in $G$. The greatest distance between any two vertices in $G$ is the $diameter$ of $G$, denoted by $diam(G)$.

We use the term \textit{unrestricted} as an attribute for a problem when there are no restrictions on the input.


\subsection{Fixed-parameter tractability}
\label{subsec:fpt}
A {\em parameterized problem} is a language $L \subseteq \Sigma^* \times \mathbb{N}$, where $\Sigma$ is a fixed, finite alphabet. For an instance $(x,k) \in \Sigma^* \times \mathbb{N}, k$ is called the {\em parameter}. A parameterized problem $L \subseteq \Sigma^* \times \mathbb{N}$ is called {\em fixed-parameter tractable} (\FPT) if there exists an algorithm $\mathcal{A}$ (called a {\em fixed-parameter algorithm}), a computable function $f: \mathbb{N} \rightarrow \mathbb{N}$, and a constant $c$ such that, given $(x,k) \in \Sigma^* \times \mathbb{N}$, the algorithm $\mathcal{A}$ correctly decides whether $(x,k) \in L$ in time bounded by $f(k)\cdot |(x,k)|^c$. The complexity class containing all fixed-parameter tractable problems is called \FPT. There is also an associated hardness hierarchy and the basic hardness classes are \W[1] and \W[2]. The {\sc clique} problem (does the given graph have a clique of size at least $k$ ?) is a canonical complete problem for \W[1] while the  {\sc dominating set} problem (does the given graph have a dominating set of size at most $k$?) is a canonical complete problem for \W[2]. We refer to~\cite{book} for more details on parameterized complexity.

Let $A, B \subseteq \Sigma^* \times \mathbb{N}$ be two parameterized problems. A {\em parameterized reduction} from $A$ to $B$ is an algorithm that, given an instance $(x,k)$ of $A$, outputs an instance $(x',k')$ of $B$ such that
\begin{enumerate}
\item 
$(x,k)$ is a YES instance of $A$ if and only if $(x',k')$ is a YES instance of $B$,
\item
$k' \leq g(k)$ for a computable function $g$, and
\item
the running time of the algorithm is $f(k)\cdot |x|^{\mathcal{O}(1)}$ for a computable function $f$.
\end{enumerate}

A \textit{polynomial compression} of a parameterized language $Q\,\subseteq\,\Sigma\times \mathbb{N} $ into a language $R\subseteq \Sigma^*$ is an algorithm that takes as input an instance $(x,k) \in \Sigma^*\times\mathbb{N}$, works in polynomial time in $|x|+k$, and returns a string $y$ such that:

$(a)$ $|y|\leq p(k)$ for some polynomial $p(.)$, and 

$(b)$ $y\in R$ if and only if $(x,k)\in Q$. 

\noindent
A {\em kernelization algorithm} is a polynomial-time algorithm that transforms an arbitrary instance of the problem to an equivalent instance (known as {\em kernel}) of the same problem, such that the size of the new instance is bounded by some computable function $g$ of the parameter of the original instance. Kernelization typically involves applying a set of rules (called {\em reduction rules}) to the given instance.
A reduction rule is a rule that translates a given instance into another. The rule is said to be {\it safe} if the reduced instance is equivalent to the original instance in the sense that the reduced instance is a YES instance if and only if the original instance is a YES instance. Unless otherwise specified, we use \textit{polynomial time} to denote a running time that is a polynomial function of the input size. 

\section{Tracking Set Systems}
\label{sec:setsystem}

In this section we study generalized versions of the \tsp problem, i.e. \tss problem. For a set system $\mathcal{P}=\{X,\mathcal{S}\}$, a \trs is a subset of elements $T \subseteq X$, that has a unique intersection with each set in the family $\mathcal{S}$ i.e. $ T \cap S_i \neq T \cap S_j$,  $\forall S_i,S_j \in \mathcal{S}$ (where $i \neq j$). For the remainder of this section, unless otherwise specified, by \trs we mean \trs for set systems.
We first consider a restricted version of \tss wherein the size of the sets in family $\mathcal{S}$ is limited to $d$, which is referred as \dts.

\subsection{$d$-Tracking Set}
\label{subsec:d-tracking-set}

In this section we give a kernel and an \textsc{FPT} algorithm for a restricted version of \tss wherein the size of each set in the family is restricted to $d$. We formally define the problem as follows.

\prob{\dts $(X,\mathcal{S},d,k)$}{A set system $(X,\mathcal{S})$, such that $\forall S\in\mathcal{S}, |S|\leq d$; parameter =$k$.}{A set ${T} \subseteq X$ where $|T|\leq k$, such that for any two distinct $S_i,S_j\in S$, it holds that $S_i\cap T\neq S_j\cap T$, if it exists.}

\noindent


When $d=2$, \dts is the same as \ivcp~\cite{henning-identifying-vertex-covers}. It is known that \ivcp is related to \textsc{Packing}, which involves finding a maximum set of disjoint packing of paths of length at least four in a graph~\cite{julien-thesis}. \textsc{Packing} is formally defined as follows.

\prob{\textsc{Packing}$(G,k)$}{A graph $G=(V,E)$.}{A maximum cardinality set $\mathcal{P}$ of paths of length at least four, such that for any two distinct paths $P_1,P_2\in \mathcal{P}$, it holds that $V(P_1)\cap V(P_2)=\emptyset$, and $\bigcup_{P\in\mathcal{P}} V(P) = V$.}

It is known from~\cite{chain-packing} that \textsc{Packing} is \NP-hard, and it is known that there exists a polynomial time reduction from \textsc{Packing} to \ivcp~\cite{henning-identifying-vertex-covers}. Hence we have the following lemma.


%


%
%

\begin{lemma}
\label{lem:dts-hard-2}
\dts is \NP-hard for $d=2$.
\end{lemma}

Consider an instance $(X,\mathcal{S},d,k)$ of \dts where $d=2$. Let $d'\geq 3$ be an integer. We introduce additional $d'-2$ dummy elements in $X$, and add those dummy elements to all the sets in the family $\mathcal{S}$. Let $(Y,\mathcal{S}',d',k)$ be the new instance obtained. All the sets in the family $\mathcal{S}'$ are of size $d'$. Since the new elements in $Y$ are common in all the sets in $\mathcal{S}'$, in order to distinguish between the sets in $\mathcal{S}'$, we necessarily need to distinguish the sets in $\mathcal{S}$ and vice-versa. Thus \dts for $d=2$ can be reduced to general \dts for any value of $d$. Further, any instance of \tss is also an instance of \dts. Hence we have the following lemma.

\begin{lemma}
\label{lem:dts-hard}
\dts (for any $d\geq 2$) and \tss are \NP-hard.
\end{lemma}

Next we give a reduction from \dts to the well known $d$-\hsp problem. For a fixed integer $d>0$, given a set system $(U,\mathcal{F})$ with each set in $\mathcal{F}$ consisting of $d$ elements, parameterized version of $d$-\hsp requires finding a hitting set of size at most $k$. 

\begin{lemma}\label{lemma:red-dts-dhs}
Let $\mathcal{P}_1=(X,\mathcal{S},d,k)$ be an instance of \dts. Then there exists an instance $\mathcal{P}_2=(U,\mathcal{F},2d,k)$ of \dhs such that $\mathcal{P}_1$ has a tracking set of size $k$ if and only if $\mathcal{P}_2$ has a hitting set of size $k$.
\end{lemma}
\begin{proof}
Let $\mathcal{P}_1=(X,\mathcal{S},d,k)$ be an instance of \dts. We construct an instance $\mathcal{P}_2= (U,\mathcal{F},2d,k)$ of \dhs as follows. Set $U=X$, and $\mathcal{F}=\{F_{RS} \mid F_{RS}=\{R\setminus S\}\cup \{S\setminus R\},\, R,S\in\mathcal{S},\,R\neq S \}$ i.e. the family consists of the symmetric difference of every pair of sets in $\mathcal{S}$. 
First we prove that if $T$ is a \trs for $\mathcal{P}_1$ then $T$ is a \hs for $\mathcal{P}_2$. Suppose not. Then there exists a set $F\in\mathcal{F}$ such that $T\cap F=\emptyset$. Due to the construction of $\mathcal{P}_2$, there exist two sets, say $R,S\in\mathcal{S}$, such that $F=\{R\setminus S\}\cup \{S\setminus R\}$. Since $T\cap F=\emptyset$, it follows that $T\cap \{\{R\setminus S\}\cup \{S\setminus R\}\}=\emptyset$ i.e. $T\cap \{R\setminus S\}= T\cap \{S\setminus R\}=\emptyset$, which implies that $T\cap R=T\cap S$. This contradicts the assumption that $T$ is a \trs for $\mathcal{P}_1$.

Next we prove that if $H$ is a \hs for $\mathcal{P}_2$ then $H$ is a \trs for $\mathcal{P}_1$. Suppose not. Then there exists two sets $R,S\in\mathcal{S}$ such that $H\cap R = H\cap S$. Thus $H\cap \{ \{R\setminus S\}\cup\{S\setminus R\}\} = \emptyset$. Due to construction of $\mathcal{F}$ it follows that there exists a set $F\in\mathcal{F}$, $F=\{R\setminus S\}\cup\{S\setminus R\}$ and $H\cap F=\emptyset$. This contradicts the assumption that $H$ is a \hs for $\mathcal{P}_2$.
Hence the lemma holds.
\end{proof}

It is known that \dhs admits a kernel with $\Oh((2d-1)k^{d-1}+k)$ sets and elements, and an FPT algorithm running in time $\Oh^*(c^k)$ where $c=d-1+\Oh(d^{-1})$~\cite{dhs-kernel}, ~\cite{dhs-fpt}. Due to this fact and Lemma~\ref{lemma:red-dts-dhs}, we have the following lemma.

\begin{lemma}\label{lemma:hs-ts-c}
\dts admits a compression of size $\Oh((4d-1)k^{2d-1}+k)$.
\end{lemma}

Observe that \dhs is in \NP, as we can verify whether a given set of elements intersects each set in the family in time polynomial in the input size. Since \dts is \NP-hard, \dhs can be reduced to \dts in polynomial time. 
Hence we have the following theorem.

\begin{theorem}\label{theorem:hs-ts}
\dts admits a polynomial kernel and an FPT algorithm running in time $\Oh^*(c^k)$ where $c$ is a polynomial function of $d$.
\end{theorem}

\subsection{Tracking Set for Set Systems}
\label{subsec:tracking-set-systems}
Although \dhs is {\sc FPT}, the general {\sc hitting set} problem is {\sc W}[2]-hard~\cite{downey}. Thus if we consider the unrestricted version of \tss, it does not help to reduce it to the {\sc hitting set} problem. Hence we consider a different problem for analysis of \tss, which is the \tcpr problem. 

We refer to an instance of the \tss as an $(x,y)$ instance if the size of the universe (element set) is $x$ and the size of the  family is $y$. 

In \tcpr we are given a set of elements $M=\{1,2,\ldots n\}$,  called vertices and a family $\mathcal{T}=\{T_1,T_2,\ldots, T_m\}$ of distinct subsets of $M$ called tests. We say that a test $T$ separates a pair $i,j$ if $|\{i,j\}\cap T|=1$. A subset $\mathcal{T}'$ of $\mathcal{T}$ is called a \tc if for every pair of distinct vertices $i,j \in M$, there exists a test $T \in \mathcal{T}'$ that separates them.
\tcpr requires finding a minimum size \tc if there exists one.

\tcpr is a well studied problem~\cite{DeBontridder2002},~\cite{DeBontridder2003}. It is known that \tcpr is \NP-hard and \APX-hard~\cite{Hall2001}. There exists an $\Oh(\log n)$-approximation algorithm for the problem~\cite{tclogn} and there is no $o(\log n)$-approximation algorithm unless $P=NP$~\cite{Hall2001}.
The parameterized complexity of \tcpr has also been studied extensively~\cite{GUTIN2013123},~\cite{Crowston2012},~\cite{crowston2016parameterizations}. 
Given $(M,\mathcal{T})$, and $k \in {\mathbb N}\cup \{0\}$, the parameterized version of \tcpr asks if there exists a \tc of size at most $k$.

For $n$ elements and a family of $m$ tests, $\lg n$ is a lower bound for the size of \tc (Theorem~\ref{testcovertheorem}$(i)$), and, $n$ and $m$ are upper bounds for the size of \tc~\cite{bondy},~\cite{crowston2016parameterizations}. 
Given lower and upper bounds of solution size, it is a natural question to ask if there exists an \FPT algorithm on a parameter $k$ which determines whether there exists a solution of size $k$ greater than the lower bound or $k$ less than the upper bound. Parameterizations of \NP-optimization problems above or below their guaranteed lower/upper bounds are well studied parameterizations~\cite{above},~\cite{abovebelow},~\cite{Krithika-above-guarantee}. 

Some results by Crowston et al.~\cite{Crowston2012} on \tcpr have been summarized in the following theorem.
\begin{theorem}~\cite{Crowston2012}
\label{testcovertheorem}
For an $(n,m)$-\tc instance, 
\begin{enumerate}[(i)]
\item
There does not exist a \tc of size less than $\lceil\lg n\rceil$. Hence 
\tcpr has a kernel of size $\mathcal{O}(2^{2^k})$, and is  fixed-parameter tractable when parameterized by solution size $k$ and can be solved in time $\mathcal{O}^*(2^{k2^k})$.
\item
Determining whether there exists a \tc of size at most $(m-k)$ is complete for the parameterized complexity class \W[1].
\item
Determining whether there exists a \tc of size at most $(n-k)$ is fixed-parameter tractable.
\item
Determining whether there exists a \tc of size at most 
$(\lg n + k)$ is hard for the parameterized complexity class \W[2].
\end{enumerate}
\end{theorem}

\tcpr is known to be a \textit{dual} of \tss~\cite{bazgan}, as explained in the following lemmas.

\begin{lemma}
\label{lemma:tc-to-tss}
Let $\{X,\mathcal{S}\}$ where $X=\{1,\cdots,n\}$ and $\mathcal{S}=\{S_1, S_2, ..., S_m\}$, be an instance of \tcpr. Then there exists an instance $\{X',\mathcal{S}'\}$ of \tss where $X'=\{x_1,\cdots,x_m \}$ and $\mathcal{S}'=\{ F_1,\cdots,F_n \}$, $F_i=\{ j \mid i \in S_j \}$ such that, there exists a \tc of size $k$ for $\{X,\mathcal{S}\}$ if and only if there exists a \trs of size $k$ for $\{X',\mathcal{S}'\}$.
\end{lemma}

\begin{lemma}
\label{lemma:tss-to-tc}
Let $\{X,\mathcal{S}\}$ where $X=\{x_1,\cdots,x_n\}$ and $\mathcal{S}=\{S_1, S_2, ..., S_m\}$, be an instance of \tss. Then there exists an instance $\{X',\mathcal{S}'\}$ of \tcpr where $X'=\{1,\cdots,m \}$ and $\mathcal{S}'=\{ F_1,\cdots,F_n \}$, $F_i=\{ j \mid x_i \in S_j \}$ such that, there exists a \trs of size $k$ for $\{X,\mathcal{S}\}$ if and only if there exists a \tc of size $k$ for $\{X',\mathcal{S}'\}$.
\end{lemma}

From Theorem~\ref{testcovertheorem} and Lemmas~\ref{lemma:tc-to-tss} and \ref{lemma:tss-to-tc}, we have the following corollary.

\begin{corollary}
\label{cor:trs-is-fpt}
For a $(n,m)$-set system the following holds.
\begin{enumerate}[(i)]
\item There does not exist a \trs of size less than $\lceil\lg m\rceil$.
\item \tss has a kernel of size $\mathcal{O}(2^{2^k})$, and is  fixed-parameter tractable when parameterized by solution size and can be solved in time $\mathcal{O}^*(2^{k2^k})$.
\item Finding a \trs of size at most $(n-k)$ is \W[1]-complete.
\item Finding a \trs of size at most $(m-k)$ is \FPT.
\item Finding a \trs of size at most $(\lg m+k)$ is \W[2]-hard.
\end{enumerate}
\end{corollary}

Gutin et al.~\cite{GUTIN2013123} have shown that there does not exist a polynomial kernel for \tcpr when parameterized by the solution size, under standard complexity theory assumptions. 
We gave a kernel for a special case of \tss in the previous subsection for the special case of \dts.

Due to Theorem~\ref{theorem:hs-ts} and Lemma~\ref{lemma:tc-to-tss}, we have the following corollary.

\begin{corollary}\label{cor:hs-tc}
For an instance of \tcpr where each element in the universe appears in at most $d$ sets in the family, there exists a polynomial kernel and an FPT algorithm running in time $\Oh^*(c^k)$ where $c$ is a function polynomial in $d$ and $k$ is the size of desired solution.
\end{corollary}

\section{Tracking Set for Paths in Graphs}
\label{sec:tssp}
In this section, we provide \FPT algorithms for \tsp problems in graphs.
In \tsp, the input is a graph $G$ with a unique source $s\in V(G)$ and a unique destination $t\in V(G)$, and the required output is a \trs, $T\subseteq V$, whose intersection with the vertex set of each \stpath is unique. 
The first problem we consider is \ddtsp where the input graph has diameter $d$, and then we tackle the general \tsp problem.

\subsection{Tracking Shortest Paths in diameter $d$ graphs}
\label{subsec:diam-d-tp}

In this section we give a kernel and an \FPT algorithm for a special case of \tsp where we consider those graphs whose diameter is at most $d$. 


\ddtsp involves finding a \trs i.e. a subset of vertices from $V(G)$, that distinguishes all \sstpaths when the input graph has diameter restricted to $d$. We define the problem formally as follows.


\pprob{\ddtsp}{An \stgraph $G$ with $diam(G)\leq d$, and an integer $k$.}{Does there exist a set $T \subseteq V(G)$ of at most $k$ vertices such that for any two \sstpaths $P_1$ and $P_2$, ${T} \cap V(P_1) \neq {T} \cap V(P_2)$ ?}

We use $(G,d,k)$ to denote an instance of the parameterized version of \ddtsp, where $G$ is a graph with $diam(G)\leq d$, and $k$ is the size of the required tracking set for tracking all \sstpaths in $G$. Observe that for an \stgraph $G=(V,E)$, if $diam(G)=2$, then $G$ consists of $s$ and $t$ being adjacent to the vertices in $V\setminus\{s,t\}$, i.e. all \stpaths in $G$ are \sstpaths, and their length is two. In such a case, all but one vertices in $V\setminus \{s,t\}$ need to be marked as trackers. Further if $dist(s,t)=2$ then $diam(G)=2$ for an \stgraph $G$.

%

Banik et al.~\cite{banik} proved that \tsp is \NP-hard when the length of shortest paths is greater than or equal to three. Note that here the graph diameter is greater than or equal to three and $dis(s,t)\geq 3$. Hence we have the following corollary.

\begin{corollary}
\label{cor:ddtsp-nph}
\ddtsp is \NP-hard when $d\geq 3$ for a fixed $d$.
\end{corollary}

%
%

Next we give a polynomial kernel and FPT algorithm for \ddtsp by reducing it to \dts. We start by giving the following reduction rule that ensures that each vertex and edge in the input graph participates in a \sstpath.

\begin{Reduction Rule}
\label{rule1}
If there exists a vertex or an edge in $G$ that does not participate in any \sstpath, delete it.
\end{Reduction Rule}

\begin{lemma}\label{lem:rule1}
Reduction Rule~\ref{rule1} is safe and can be applied in polynomial time.
\end{lemma}
\begin{proof}
Let $G=(V,E)$ be a graph, where $|V|=n$ and $|E|=m$. If a vertex or an edge does not participate in any \sstpath in $G$, it cannot play a role in tracking shortest \stpaths in $G$. To implement the rule, we first find the distance between $s$ and $t$ in $G$, using a breadth first search (BFS). Let $l$ be the length of a shortest \stpath in $G$. Now for each edge $e=ab\in E$, we check if,

$dis(s,a) + dis(b,t) + 1 = l$ or $dis(s,b) + dis(a,t) + 1 = l$.

\noindent
If above condition is not satisfied, then we remove the edge $e$ from $G$. This step takes $\Oh(m(n+m))$ time.
After above step, we also remove all isolated vertices from $G$, in $\Oh(n)$ time. 
\end{proof}

The main challenge in solving \ddtsp using \dts is that while the sets that need to be distinguished are received as a part of the input in \tss, the \sstpaths, which need to be distinguished are implicit in the input for \tsp. Thus, we need a procedure to procure the family of \sstpaths from $G$ for solving \tsp. 

Although for a general graph, counting the number of \stpaths is hard for the complexity class $\#P$~\cite{valiant1979complexity}, for some special class of graphs it can be done in polynomial time. Particularly counting \sstpaths for a graph can be done in polynomial time $\Oh(m+n)$ as explained below.

In order to construct the set system $(U,\mathcal{F})$ we first define \textit{level} $L(v)$ of a vertex $v\in V(G)$ as the length of the shortest path from $s$ to $v$.  
After the application of Reduction Rule~\ref{rule1}, there does not exist an edge between vertices equidistant from $s$ (or $t$). In fact, the endpoints of each edge are such that the difference between their distances from $s$ (or $t$) is always exactly one. Thus the vertices of the graph can be categorized into layers, such that each layer consists of the vertices equidistant from $s$ (or $t$). 
Such a graph is called a \textit{layered $s$-$t$ graph}. Next we have the following observation that helps to enumerate all \sstpaths in a layered $s$-$t$ graph. $V_\ell$ is used to denote the set of vertices at level $\ell$. For a path $P$, we use $P\cdot\{v\}$ to denote a path formed by concatenating the path $P$ with vertex $v$, given that $v$ is a neighbor of one of the end points of $P$.

\begin{Observation}
\label{obs:enum-sstpaths}
For a vertex $v\in G$, let $\mathcal{P}(s,v)$ be the set of shortest paths from $s$ to $v$. If $v\in V_\ell$, then $\mathcal{P}(s,v)=\bigcup_{u\in N(v)\cap V_{\ell-1}} \{ P\cdot\{v\} \mid P\in\mathcal{P}(s,u)\}$ is the set of shortest paths from $s$ to $v$.
\end{Observation}
The above observation can be used iteratively, starting from $s$ and going level by level till $t$, to enumerate all \sstpaths in $G$. Hence we have the following lemma.

\begin{lemma}\label{lemma:count-sstp}
The family $\mathcal{F}$ of the vertex set of all \sstpaths can be enumerated in $\Oh(|\mathcal{F}|(m+n))$ time with a polynomial delay.
\end{lemma}


Let $(G,d,k)$ be an instance of \ddtsp.
We can create an instance $(X,\mathcal{S},d',k')$ of \dts from $G$ as follows. We introduce an element in $X$ for each vertex in $G$. The vertex set of each \sstpath in $G$ forms a set in the family $\mathcal{S}$. We can construct the family $\mathcal{S}$ using Observation~\ref{obs:enum-sstpaths}. Let $d'=dis(s,t)$. Since $diam(G)=d$, the length of a \sstpath in $G$ will be less than or equal to $d$, i.e. $d'\leq d$. It can be seen that there exists a \trs of size $k$ in $G$ if and only if there exists a \trs of size $k'=k$ for $(X,\mathcal{S},d',k')$. Note that this can be generalized for the case when the graph diameter and the corresponding set sizes in a set system are unbounded.
Hence we have the following lemma.

\begin{lemma}
\label{lemma:ddtsp-dts}
Let $(G,d,k)$ be an instance of \ddtsp. Then there exists an instance $(X,\mathcal{S},d',k)$ of \dts such that $(G,d,k)$ is a YES instance if and only if $(X,\mathcal{S},d',k)$ is a YES instance. Further, if $(G,k)$ is an instance of \tsp, then there exists and equivalent instance $(X,\mathcal{S},k)$ of \tss.
\end{lemma}
%

If the time taken to enumerate all \sstpaths is \FPT then the reduction from \ddtsp to \dts (or \tsp to \tss) can be done in \FPT time as stated in Lemma~\ref{lemma:count-sstp}. Hence we have the following corollary.

\begin{corollary}
\label{cor:dtsp-fpt}
$(i)$ \ddtsp admits a polynomial kernel and an FPT algorithm running in time $\Oh^*(c^k)$ where $c$ is a function of $d$. 

$(ii)$ \tsp admits a kernel of size $\Oh(2^{2^k})$ and can be solved with an FPT algorithm running in time $\Oh^*(2^{k2^k})$.
\end{corollary}
\begin{proof}
Let $G$ be the graph in the input instance of \ddtsp or \tsp.
We start by applying Reduction Rule~\ref{rule1} to $G$. Using Lemma~\ref{lemma:count-sstp} we can enumerate the vertex sets of all \sstpaths in $G$. If there are more than $2^k$ \sstpaths in $G$, then due to Corollary~\ref{cor:trs-is-fpt}$(i)$, it is a NO instance. Else we proceed as follows.

In the case of \ddtsp, first we use Lemma~\ref{lemma:ddtsp-dts} to reduce it to an equivalent instance of \dts. From Lemma~\ref{lemma:red-dts-dhs}, we know that \dts can be reduced to \dhs. Hence, due to~\cite{dhs-kernel}, we have that \ddtsp admits a polynomial kernel and an FPT algorithm running in time $\Oh^*(c^k)$ where $c$ is a function of $d$.

In case the input is an instance of \tsp, we use Lemma~\ref{lemma:ddtsp-dts} to reduce it to an equivalent instance of \tss. Then we use Corollary~\ref{cor:trs-is-fpt} to give a kernel and FPT algorithm. Hence, \tsp admits a kernel of size $\Oh(2^{2^k})$ and can be solved with an FPT algorithm running in time $\Oh^*(2^{k2^k})$.
\end{proof}

\subsection{Improved FPT algorithm for \tsp}
\label{subsec:improvedfpt}

Here we obtain an improved \FPT algorithm for \tsp. This is done by first reducing \tsp to the problem of tracking all paths in a directed acyclic graph. Some additional preprocessing rules are given for DAGs, that result in a larger lower bound for the number of \stpaths in a graph, thereby giving a smaller upper bound for the size of the vertex set. 

%

Let $(G,k)$ be an instance of \tsp. We start by applying Reduction Rule~\ref{rule1}. This removes all those vertices and edges from $G$ that do not participate in any \sstpath.

Next we reduce \tsp to the problem of tracking all \stpaths in a directed acyclic graph. We formally define the problem as follows.

\prob{\tpdag}{A directed acyclic \stgraph $G=(V,E)$.}{A minimum set of vertices ${T} \subseteq V$, such that for any two distinct \stpaths $P_1$ and $P_2$ in $G$, it holds that ${T} \cap V(P_1) \neq {T} \cap V(P_2)$.}

An instance of the parameterized version of \tpdag is denoted by $(G,k)$, where $G$ is the input graph and $k$ is the size of the desired tracking set for $G$. Note that a pair of paths in $G$ cannot have the same vertex set but different sequence of vertices, as this would create a cycle. Hence in a DAG, in order to identify each \stpath uniquely, it is sufficient for each \stpath to have a unique intersection with a tracking set.
Next we prove that there exists a polynomial time reduction from \tsp to \tpdag.

\begin{lemma}\label{lemma:red-tsp-tpdag}
Let $(G,k)$ be an instance of \tsp. Then there exists an instance $(G',k)$ of \tpdag such that $(G,k)$ is a YES instance if and only if $(G',k)$ is a YES instance. 
\end{lemma}
\begin{proof}
We assume $G$ to be preprocessed by Reduction Rule~\ref{rule1}. We create the graph $G'$ from $G$ as follows. Construct the \stgraph $G'$ by creating a copy of $G$. Next, orient each edge in $G'$ towards the destination $t$. 
Note that now each \sstpath in $G$ is an \stpath in $G'$. Further, each \stpath in $G'$ is a \sstpath in $G$. This holds due to the application of Reduction Rule~\ref{rule1} on $G$. Since the set of \sstpaths in $G$ is same as the set of \stpaths in $G'$, it holds that if there exists a tracking path $T$ of size $k$, for all \sstpaths in $G$, then $T$ is also a tracking set for all \stpaths in $G'$.
\end{proof}

Lemma~\ref{lemma:red-tsp-tpdag} also proves the hardness for the problem of tracking all \stpaths in directed acyclic graphs.

\begin{corollary}
\tpdag is \textsc{NP-}hard.
\end{corollary}

In the remainder of the section, by graph we mean a DAG, and by path we mean a directed path. Now we give a kernel and an FPT algorithm for \tpdag. Note that here our objective is to track all \stpaths in a DAG.
We start by giving a reduction rule that removes all those vertices from a DAG $G$, that do not participate in any \stpath in $G$.

\begin{Reduction Rule}
\label{red:useful-dag}
If there exists a vertex or an edge in $G$ that does not participate in any \stpath, delete it.
\end{Reduction Rule}

\begin{lemma}
\label{lemma:red-useful-dag}
Reduction Rule~\ref{red:useful-dag} is safe and can be implemented in polynomial time.
\end{lemma}
\begin{proof}
Let $G$ be a DAG. If a vertex or an edge does not participate in any \stpath in $G$, it cannot play a role in tracking \stpaths in $G$. Hence the rule is safe.

In order to implement the rule, we perform the following steps exhaustively:
\begin{enumerate}
\item Delete all incoming edges on the source $s$.

\item Delete all outgoing edges from the destination $t$.

\item For a vertex $v\in V(G)\setminus\{s,t\}$, if $deg^+(v)=0$ or $deg^-(v)=0$, then delete $v$ along with all its incident edges.
\end{enumerate}

Note that after performing the above steps, there exists a path from $s$ to each vertex in $G$, and there exists a path from each vertex in $G$ to $t$. Suppose not. Let $\pi$ be a topological ordering of $G$. Let $x$ be a vertex that is not reachable from $s$. Without loss of generality, let $x$ be the first vertex in $\pi$, such that $x$ is not reachable from $s$. Thus all vertices that appear before $x$ in $\pi$ are reachable from $s$. Since $deg^-(x)\geq 1$, there exists a vertex $y$ such that $y\in N^-(x)$. Since $y$ is an in-neighbor of $x$, $y$ is reachable from $s$. Further, since $(y,x)\in E(G)$, it holds that $x$ is also reachable from $s$. Similarly it can be proven that $t$ is reachable from each vertex in $G$. Note that for a vertex $v\in G$, a path from $s$ to $v$ cannot intersect a path from $v$ to $t$ at any vertex other than $v$, as this would create a cycle, and contradict the fact that $G$ is a DAG. Hence now each vertex and edge in $G$ participates in an \stpath. It can be seen that the total time taken to apply the rule is $\Oh(n+m)$.
\end{proof}

Note that after application of Reduction Rule~\ref{red:useful-dag}, each vertex in a graph except $s$ and $t$ has non zero in-degree and out-degree. Further, if the degree of a vertex is two, then both its out-degree and in-degree are exactly one.
For the remainder of the paper we assume that the graph has been preprocessed using Reduction Rule~\ref{red:useful-dag}.

Next we give a reduction rule that ensures that the degree of $s$ and $t$ is at least two.

\begin{Reduction Rule}
\label{red:deg-one-st}
If $deg(s)=1$ and $u\in N^+(s)$, then delete $s$ and set $s=u$. If $deg(t)=1$ and $v\in N^-(t)$, then delete $t$ and set $t=v$.
\end{Reduction Rule}

\begin{lemma}
\label{lemma:red-deg-one-st}
Reduction Rule~\ref{red:deg-one-st} is safe and can be implemented in polynomial time.
\end{lemma}
\begin{proof}
Observe that if $deg(s)=1$ and $u\in N^+(s)$, then all paths that start at $s$, pass through $u$ and vice-versa. Similarly, if $deg(t)=1$ and $v\in N^-(t)$, then all paths that reach $t$, pass through $v$ and vice-versa. Hence in such a case, it is safe to assign the neighbor of $s$ ($t$) as the source (destination), and delete the original $s$ ($t$). It can be seen that the rule can be applied in constant time.
\end{proof}

If after applying reduction rules, the graph becomes a singleton, then we return a YES. Else, henceforth we assume that the reduced graph is not a singleton.
Next we give a lemma that gives a lower bound for the number of \stpaths in a directed acyclic graph reduced using Reduction Rules~\ref{red:useful-dag} and \ref{red:deg-one-st}.



\begin{lemma}
\label{lemma:paths}
In a graph $G$ reduced by Reduction Rule~\ref{red:useful-dag}, the number of \stpaths is at least $1+\sum\limits_{v\in V\setminus\{t\}} (deg^+(v) -1)$.
\end{lemma}

\begin{proof}
Let $G$ be a graph preprocessed using Reduction Rule~\ref{red:useful-dag}, i.e. each vertex and edge in $G$ participates in a \stpath.
Let $p=1+\sum\limits_{v\in V\setminus\{t\}} (deg^+(v) -1)$. The proof is by induction on $p$. The base case is when $p=1$. This is possible only when $\sum\limits_{v\in V\setminus\{t\}} (deg^+(v) -1)=0$, which implies that the out-degree of all vertices in $V\setminus\{t\}$ is one, and hence the graph is a single path between $s$ and $t$. Hence the claim holds.

Assume that the claim is true for $p\leq k$, where $k\geq 2$. Consider $p=k+1$, i.e. $1+\sum\limits_{v\in V\setminus\{t\}} (deg^+(v) -1)=k+1$, where $k\geq 2$. 



In the following, we will reduce the value of $p$ by exactly one, and show that the number of \stpaths in the graph is also reduced by at least one. Let $x$ be a vertex closest to $t$, such that $deg^+(x)\geq 2$. Such a vertex exists, since due to Reduction Rule~\ref{red:deg-one-st}, $deg^+(s)\geq 2$.
Due to Reduction Rule~\ref{red:useful-dag}, there exists a directed path, say $P_{xt}$, from $x$ to $t$. Let $v$ be the first vertex in $P_{xt}$ such that $deg^-(v)\geq 2$. Such a vertex exists since due to Reduction Rule~\ref{red:deg-one-st}, $deg^-(t)\geq 2$. Let $P'_{xt}$ be the subpath of $P_{xt}$ lying between vertices $x$ and $v$, excluding the vertices $x$ and $v$. Note that $P'_{xt}$ is either a single edge or a path of degree two vertices. Let $G'$ be the graph obtained after the deletion of $P'_{xt}$. Note that the out-degree of $x$ is reduced by one in $G'$. For each vertex deleted in $P'_{xt}$, the value of $p$ remains unchanged as the reduction in summation of out-degree is accompanied by an equal reduction in the count of vertices. Hence, in $G'$, $p$ is reduced by exactly one, i.e. $p=k$. Further note that, after the deletion of $P'_{xt}$, each vertex and edge in the graph still participates in an \stpath. Hence by induction hypothesis, the claim holds for $p=k$. Observe that the deletion of $P'_{xt}$ reduces the number of \stpaths by at least one. Hence the claim holds for $p=k+1$ as well.
This completes the proof.
\end{proof}

Next we give a reduction rule that helps remove long degree two paths (paths containing only vertices with degree two in the graph) from the input graph.

\begin{Reduction Rule}
\label{red:deg-2-dag}
In a graph $G$, if there exist $x,y,z\in V(G)$, and $(x,y),(y,z)\in E(G)$, and $deg(x)=deg(y)=2$, then delete the vertex $y$ and introduce the edge $(x,z)$ in $G$.
\end{Reduction Rule}

\begin{lemma}
\label{lemma:red-deg-2-dag}
Reduction Rule~\ref{red:deg-2-dag} is safe and can be applied in polynomial time.
\end{lemma}

\begin{proof}
Let $(G,k)$ be an instance of \tpdag.
Let $(x,y),(y,z)\in E(G)$ and  $deg(x)=deg(y)=2$.
Consider the possibility when there already exists an edge between $x$ and $z$ in $G$. If $(x,z)\in E(G)$, then $deg^-(x)=0$, which is not possible due to Reduction Rule~\ref{red:useful-dag}. If $(z,x)\in E(G)$, then $x,y,z$ induce a cycle in $G$, which contradicts the fact that $G$ is a DAG. Hence if $deg(x)=deg(y)=2$, and $(x,y),(y,z)\in E(G)$, then there cannot exist an edge between $x$ and $z$ in $G$.

Further, since $y$ participates in an \stpath if and only if $x$ participates in that path, if $y$ needs to be marked as a tracker, $x$ can replace it as a tracker. Hence, the reduction rule is safe.

In order to apply the rule, we consider each vertex $u\in V(G)$. If $deg(u)=2$ and $deg(v)=2$, where $v\in N^+(u)$, then we delete $v$ and introduce an edge between $u$ and $w\in N^+(v)$. This can be done in $\Oh(n+m)$ time. 
\end{proof}

\noindent
After the application of Reduction Rule~\ref{red:deg-2-dag}, we have the following observation.

\begin{Observation}
\label{lem:halfhighdegree}
For a pair of vertices $u,v\in V(G)$ there exists at most one vertex of degree two that is adjacent to both $u$ and $v$. 
\end{Observation}

Next we give a lower bound for the number of \stpaths in a DAG reduced under Reduction Rules~\ref{red:useful-dag}, \ref{red:deg-one-st}, and \ref{red:deg-2-dag}. We call such DAGs \textit{reduced DAGs}.

\begin{lemma}
In a reduced DAG $G$ on $n$ vertices there exists at least $n/5$ \stpaths.
\label{lemma:root-n-paths}
\end{lemma}
\begin{proof}
Let $p$ be the number of paths in $G$ and $m$ be the number of edges in $G$. From Lemma~\ref{lemma:paths}, it is known that $p \geq 1 +\sum\limits_{v\in V\setminus\{t\}} (deg^+(v) -1)$. This implies that $p\geq m-n+2$. 
Let $n_2$ be the number of vertices with degree exactly two and $n_3$ be the number of vertices with degree at least three in $G$. Note that due to Reduction Rules~\ref{red:useful-dag} and \ref{red:deg-one-st}, there does not exist a vertex of degree one in $G$. 

Consider a graph $G'$ that is obtained from $G$ by short-circuiting all the vertices with degree two. 
Note that the \stpaths in $G'$ are same as those in $G$, except that some of the paths in $G$ may have additional degree two vertices on them. Thus, the number of \stpaths in $G'$ is also $p$. Let $m'$ be the number of edges in graph $G'$. Observe that $s$ and $t$ are the only vertices that can not be short-circuited if they are of degree two. Further, no two vertices of degree two (except for $s$ and $t$) can be adjacent in $G$, due to Reduction Rule~\ref{red:deg-2-dag}. Hence we have 
\begin{align}
n_2-2\leq m' \label{eq:1}
\end{align}
In graph $G'$, the degree of all vertices other than $s$ and $t$ is at least three, and the degree of $s$ and $t$ is at least two. Thus, 
\begin{align}
m'\geq (3(n_3-2) + 2 +2)/2 = 1.5n_3-1 \label{eq:2}
\end{align}

Since $m=m'+n_2$ and $n=n_3+n_2$, we have
\begin{align}
p\geq m-n+2 = m' - n_3 + 2 \label{eq:3}
\end{align}
From Equations~\ref{eq:2} and \ref{eq:3}, we have
\begin{align}
p\geq n_3/2 + 1 \label{eq:4}
\end{align}
Therefore,
\begin{align*}
n &= n_3 + n_2 \\
&\leq n_3 + m' + 2 \textrm{ (due to Equation~\ref{eq:1})} \\
&\leq p + 2n_3  \textrm{ (due to Equation~\ref{eq:3})} \\
&\leq p + 4(p-1)  \textrm{ (due to Equation~\ref{eq:4})} \\
&= 5p-4.
\end{align*}
Hence $p\geq n/5$.
\end{proof}

Next we have the following observation which helps to count the number of \stpaths in a DAG, similar to Observation~\ref{obs:enum-sstpaths}.

\begin{Observation}
\label{obs:count-paths}
For a vertex $v\in V(G)$, the number of paths from $s$ to $v$ is denoted by $p_{sv}$. The number of paths from $s$ to $v$ is equal to the sum of number of paths from $s$ to each of the in-neighbors of $v$, i.e. $p_{sv}=\sum_{u\in N^-(v)} p_{su}$. Hence the number of \stpaths in $G$ is equal to $\sum_{u\in N^-(t)} p_{su}$.
\end{Observation}
Observation~\ref{obs:count-paths} gives a recursive algorithm to compute the number of \stpaths in $G$ in $\Oh(m+n)$ time, where $m$ is the number of edges and $n$ is the number of vertices in $G$.

Next we give a condition that helps verify if a set of vertices is a \trs for all \stpaths in a graph, in polynomial time. The condition was introduced in~\cite{caldam2018}, but we re-state it here and show that it holds for DAGs as well. 

\begin{Tracking Set Condition}
\label{tssp-condition}
For a graph $G$, a set of vertices $T\subseteq V(G)$ is said to follow the \tsc if there exists at most one path between any two vertices $u,v \in T \cup \{s,t\}$ in the graph $G(V\setminus (T \setminus \{u,v\}))$.
\end{Tracking Set Condition}

Next we show that \textbf{Tracking Set Condition} is necessary and sufficient for a set of vertices to be a \trs. 

\begin{lemma}
\label{lem:verif}
Let $G=(V,E)$ be a DAG with source $s$ and destination $t$, $s,t\in V$. If each vertex and edge in $G$ participates in an \stpath and $T \subseteq V$ is a set of vertices, then $T$ is a \trs for $G$ if and only if $T$ follows the \tsc.
\end{lemma}

\begin{proof}
Let $T\subseteq V$ be a \trs for $G$. We claim that $T$ follows the \tsc. Suppose not. Then there exists two vertices $u,v\in T\cup\{s,t\}$ such that there exists two paths, say $P_1,P_2$, between $u$ and $v$ that do not contain any vertex from $T\setminus \{u,v\}$. Due to Reduction Rule~\ref{red:useful-dag}, each vertex in $G$ participates in a \stpath. Hence there exists a path from $s$ to $u$, say $P_{su}$, and there exists a path from $v$ to $t$, say $P_{vt}$. Note that since $G$ is a directed acyclic graph, $P_{su}$ can intersect with $P_1$ and $P_2$ only at $u$. Similarly, $P_{vt}$ can intersect with $P_1$ and $P_2$ only at $v$. Observe that paths $P_{su}\cdot P_1 \cdot P_{vt}$ and $P_{su}\cdot P_2\cdot P_{vt}$ are two distinct \stpaths that contain the same set of trackers. This contradicts the assumption that $T$ is a \trs for $G$.

Conversely, let $T\subseteq V$ be a set of vertices that follows the \tsc. We claim that $T$ is a \trs for $G$. Suppose not. Then there exists two distinct \stpaths in $G$, say $P_1,P_2$, that contain the same set of trackers. Let $T'=T\cup\{s,t\}$. Let $x$ be the first vertex on $P_1$ such that $x\in V(P_1)\cap V(P_2)\cap T'$. Let $y$ be the first vertex on $P_1$ after $x$, such that $y\in V(P_1)\cap V(P_2)\cap T'$. Since the graph is not a singleton, $P_1$ and $P_2$ share at least two vertices, hence $x,y$ exist. Note that if $P_1$ and $P_2$ are vertex disjoint paths except for vertices $s$ and $t$, then vertices $x$ and $y$ are $s$ and $t$. Let $P_1'$ be the subpath of $P_1$ between vertices $x$ and $y$, and $P_2'$ be the subpath of $P_2$ between vertices $x$ and $y$. Observe that $x,y$ is a pair in $T'$ such that there exists two paths between $x$ and $y$ that do not contain any vertices from $T\setminus\{x,y\}$. This violates the \tsc and thus contradicts the assumption that $T$ follows \tsc.
\end{proof}

Hence for a graph $G=(V,E)$, where $|V|=n$ and $|E|=m$, for a set of vertices $T\subseteq V$, $|T|\leq k$, it can be verified whether $T$ is a \trs for $G$ in $\mathcal{O}(k^2 (m+n))$ time, 
by checking that there is no more than one path between every pair of vertices in $T\cup\{s,t\}$.

\begin{theorem}
\label{theorem:tssp}
Let $(G,k)$ be an instance of \tpdag, where $G$ is a graph on $n$ vertices and $m$ edges. Then there exists an \FPT algorithm running in time $\mathcal{O}(2^{k^2+3k} k^2 (m+n))$ that decides whether $(G,k)$ is a YES instance or not.
\end{theorem}
\begin{proof}
We start by applying Reduction Rules~\ref{red:useful-dag} and \ref{red:deg-2-dag}. For convenience, we use $(G,k)$ to denote the reduced instance, and $n$ and $m$ to denote the number of vertices and edges in $G$. Let $p$ be the number of \stpaths in $G$. In order to track $p$ paths, we need at least $\lg(p)$ trackers (follows from Corollary~\ref{cor:trs-is-fpt}$(i)$).
From Lemma~\ref{lemma:root-n-paths}, we know that $p\geq n/5$. Hence $\lg(p) \geq \lg (n/5) $. 
Using Observation~\ref{obs:count-paths}, we find the value of $p$ in $\Oh(m+n)$ time. Next, if $k<\lg(p)$, i.e. $k< \lg (n/5) $, we report that it is a NO instance. Else, $k\geq\lg (n/5)$. Hence $n \leq 5(2^k)$.
Now for each subset of $T\subseteq V$ of size $k$, we verify whether $T$ is a \trs for $G$, using the \tsc in $\mathcal{O}(k^2 (m+n))$ time.
Thus in $\mathcal{O}(2^{k^2+3k}{k^2}(m+n))$ time, we can find a \trs of size at most $k$ if one exists.
\end{proof}

%

\section{Acknowledgement}
We thank the anonymous referees for comments that helped improve the overall presentation of the paper and the bound in Lemma~\ref{lemma:root-n-paths}.

\section{Conclusions} 
\label{sec:conclusion}


In this paper we have studied \trs problems for set systems, \sstpaths in undirected graphs and \stpaths in DAGs.
We gave a polynomial kernel for \tss for the case when size of the sets in the family is restricted to at most $d$. 
The improved kernel and algorithm for \tss in this case also implies corresponding improvements for \tcpr for the case when frequency of appearance of each element is restricted to at most $d$ sets.  

The results for \tss are then used to give an FPT algorithm for \tsp in graphs, and a polynomial kernel for the case when the diameter of the input graph is restricted to $d$. Finally we give an improved algorithm for \tsp by first reducing it to \tpdag and then using some structural properties of DAGs.


Possible directions of further study include exploration of other variants of \tss and obtaining improved FPT algorithms for \tsp in special graph classes.


\bibliography{ref1}

\end{document}